\documentclass[11pt]{article}

\usepackage{amssymb,amsmath,amsthm,setspace}
\usepackage[top=2cm, bottom=2cm, left=2cm, right=2cm]{geometry}
\usepackage[usenames,dvipsnames]{color}
\newcommand{\url}{}
\usepackage[title,titletoc]{appendix}

\theoremstyle{plain}

\newtheorem{thm}{THEOREM}[section]
\newtheorem{lm}[thm]{LEMMA}

\newtheorem{prop}[thm]{PROPOSITION}
\theoremstyle{definition}

\theoremstyle{definition}
\newtheorem{remark}[thm]{Remark}

\newcommand{\R}{{\mathord{\mathbb R}}}

\newcommand{\cA}{{\mathord{\cal A}}}
\newcommand{\cB}{{\mathord{\cal B}}}

\renewcommand{\|}{{\Vert}}
\newcommand{\q}{\mathbf{q}}
\renewcommand{\v}{\mathbf{v}}
\newcommand{\E}{\mathbf{E}}
\newcommand{\F}{\mathbf{F}}
\renewcommand{\j}{\mathbf{j}}

\def\dd{\mathrm{d}}

\def\X{\mathbf{X}}
\def\V{\mathbf{V}}
\def\v{\mathbf{v}}
\def\W{\mathbf{W}}
\def\w{\mathbf{w}}
\def\Q{\mathbf{Q}}
\def\q{\mathbf{q}}
\def\n{\mathbf{n}}
\def\m{\mathbf{m}}
\def\cC{\mathcal{C}}

\numberwithin{equation}{section}
\begin{document}

\markboth{\scriptsize{BCELM\today}}{\scriptsize{BCELM \today}}

\title{{P}ropagation of Chaos for a Thermostated Kinetic Model}

\author{\vspace{5pt} 
F. Bonetto$^1$, E.A. Carlen$^{2}$, R. Esposito$^{3}$, J.L. Lebowitz$^{4}$ and R. Marra$^{5}$\\
\vspace{5pt}\small{$1.$ School of Mathematics, Georgia Institute of Technology,}\\[-6pt]
\small{
Atlanta GA 30332 USA}\\
\vspace{5pt}\small{$2,3.$ Department of Mathematics, Hill Center, Rutgers University,}\\[-6pt]
\small{
110 Frelinghuysen Road, Piscataway NJ 08854-8019 USA}\\
\vspace{5pt}\small{$4.$ International Research Center M\&MOCS, Univ. dell'Aquila,}\\[-6pt]
\small{
Cisterna di Latina, (LT) 04012 Italy}\\
\vspace{5pt}\small{$5.$ Dipartimento di Fisica and Unit\`a INFN, Universit\`a di Roma Tor Vergata}\\[-6pt]
\small{
100133 Roma, Italy}\\
}
\maketitle
\footnotetext                                                                         
[2]{Work partially
supported by U.S. National Science Foundation
grant DMS 1201354.    }                           
\footnotetext
[3]{Work partially
supported by U.S. National Science Foundation
grant PHY 0965859.}
\footnotetext
[5]{Work partially
supported by MIUR and GNFM-INdAM.\\
\copyright\, 2012 by the authors. This paper may be
reproduced, in its
entirety, for non-commercial purposes.}

\medskip

\centerline{\it Dedicated to Herbert Spohn in friendship and appreciation}
\medskip

\begin{abstract}
 We consider a system of $N$ point particles moving on a $d$-dimensional torus 
$\mathbb{T}^d$. Each particle is subject to a uniform field $\E$ and random 
speed conserving collisions $\v_i\to\v_i'$ with $|\v_i|=|\v_i'|$. This model is 
a variant of the   Drude-Lorentz  model of electrical conduction \cite{AM}. In 
order to avoid heating by the external field, the particles also interact with a 
Gaussian thermostat which keeps the total kinetic energy of the system constant. 
The thermostat induces a mean-field type of interaction between the particles. 
Here we prove that, starting from a product measure,  in the limit $N\to\infty$, 
the one particle velocity distribution $f(\q,\v,t)$ satisfies a self consistent 
Vlasov-Boltzmann equation, for all finite time $t$.   This is a consequence of 
``propagation of chaos'', which we also prove for this model. 
\end{abstract}

\section{Introduction}

The derivation of autonomous kinetic equations describing the time evolution of 
the one particle phase-space distribution function $f(\q,\v,t)$ in a macroscopic 
system of $N$ interacting particles, $N\gg1$, is a central problem in 
non-equilibrium statistical mechanics. Examples of such equations are the 
Boltzmann (BE), Vlasov (VE) and Boltzmann-Vlasov (BVE) equations \cite{Neu}. All 
of these equations have a quadratic non-linearity, and their derivation from the 
$N$-particle microscopic dynamics requires proving (or assuming) that the two 
particle distribution function, $f^{(2)}(\q,\q',\v,\v',t)$ factorizes (in an 
appropriate sense) for times $t$, $0<t<T$, as a product $f(\q,\v,t)f(\q',\v',t)$ 
when $N\to\infty$ given that the $N$-particle distribution has a factorized form 
at $t=0$. Such a property of the $N$-particle dynamics, which goes under the 
name of ``propagation of chaos'' (POC),  was first proved by Marc Kac for a  
spatially homogeneous stochastic model leading to a BE \cite{Kac}. In his model, 
one picks a pair of particles at random and changes their velocities 
$(\v,\w)\mapsto (\v',\w')$ as if they had undergone an energy conserving 
collision. Kac's proof makes essential use of a boundedness property of the 
generator of the master equation (Kolmogorov forward equation) corresponding to 
the Kac process. This boundedness property is valid for a {\em Maxwellian 
collision model} in which the rate of collision is independent of the speed of 
the particles. It is not valid for hard-sphere collisions, for which the rate 
depends on single-particle speeds, which can become arbitrarily large with large 
$N$ even when the energy per particle is of unit order. 

Propagation of chaos for the Kac model with hard sphere collisions was studied 
by McKean and his student Grunbaum \cite{AG}.  The first complete proof was 
given by Sznitman \cite{SA} in 1984. These authors used a  more probabilistic 
methodology, and avoided direct analytic estimates on the master equation. This 
method has been extended to more general collisions and to higher dimensions by 
many authors: for the latest results, which are {\em quantitative},  see 
\cite{MM}. This method was also used by Lanford in deriving the BE, for short 
times $T$, for a system model of hard spheres evolving under Newtonian dynamics, 
in the Boltzmann-Grad  limit, see \cite{La} and book and articles by H. Spohn 
\cite{Sp,Spb}.

The derivation of the VE for a weakly interacting Hamiltonian system was first 
done by Braun and Hepp \cite{BrH}. They  considered a system of $N$-particles in 
which there is a smooth force on each particle that is a sum of contributions 
from $O(N)$ particles with each term being of $O(1/N)$, see also Lanford 
\cite{La}. There is then, in the limit $N\to\infty$, an essentially 
deterministic smooth force on a given particle and the evolution of the one 
particle distribution is then described by the VE. 

There are no rigorous derivations of the VBE where, in addition to the smooth 
collective force, there are also discontinuous collisions.  Actually, even the 
appropriate scaling limit in which the VBE could be derived is not obvious: the 
collision part is obtained in the Grad-Boltzmann limit which means low density; 
the Vlasov force is obtained in the mean field limit. The two regimes could be 
incompatible. In \cite{ba},  it is shown that this is not the case, by proposing 
a scaling limit in which the VBE is derived formally from the BBGKY hierarchy.

In this paper we go some way towards deriving such a VBE.  The system we study 
is a variation of the Drude-Lorentz model of electrical conduction. In this 
model, the electrons are described as point particles moving in a constant 
external electric field $\E$ and undergoing energy conserving collisions with  
much heavier ions. In order to avoid the gain of energy by the electrons from 
the field, we use a Gaussian thermostat that fixes the total kinetic energy of 
the $N$-particle system.

A model in which one particle moves on a two dimensional torus with fixed disc 
scatterers subject to a field $\E$ and a Gaussian thermostat was originally 
introduced by Moran and Hoover \cite{MH}.  Exact results for this system were 
derived in \cite{CELS}. More information, both analytic and numerical, on this 
system for $N=1$ can be found in \cite{BCKL, BCKLs}. There are no rigorous 
results for the deterministic case when $N>1$, which was first studied in 
\cite{BDLR}, and further studied in \cite{BCKL2}

To obtain a many-particle model model that is amenable to analysis, it is 
natural to replace the deterministic collisions with the fixed disc scatterer by 
stochastic collisions. This too was done in \cite{BDLR} which  introduced  a 
variant of the model in which the deterministic collisions are replaced by 
``virtual collisions''  in which each  particle undergoes  changes in the 
direction of its velocity, as in a collision with a disc scatterer, but now with 
the redirection of the velocity being random, and the ``collision  times'' 
arriving in a Poisson stream, with a rate proportional to the particle speed. 
This variant is amenable to analysis for $N>1$, as shown in  \cite{BDLR}.  In 
particular, \cite{BDLR}  presented an heuristic  derivation of a BE for  this 
many-particle system. To make this derivation rigorous, one must prove POC for 
the model, which is the aim of the present paper.

In this class of models the collisions  involve only one particle at a time, and 
are not the source of any sort of correlation. The particles interact only 
through the Gaussian thermostat. Dealing with this is naturally quite different 
from dealing with the effects of binary collisions. On a technical level, the 
novelty of the present paper is to provide means to deal with the correlating 
effects of the Gaussian thermostat. To keep the focus on this issue, we make one 
more simplifying assumption on the collisions: In the model we treat here, the 
virtual collision times arrive in a Poisson stream that is independent of 
particle speeds. In other words, our virtual collisions are ``Maxwellian''.  
However, one still cannot use Kac's strategy for proving POC in this case since 
the terms in the generator of the process representing the Gaussian thermostat 
involve derivatives, and are unbounded. Because the virtual collisions are not 
responsible for any correlation between the particles, the assumption of 
Maxwellian collisions is for technical reasons only. Later in the paper, when it 
is clear how we are using the Maxwellian property, we shall explain how we 
expect to be able to build on our present results and treat head-sphere virtual 
collisions as well. Interesting generalizations of the Gaussian thermostat 
problem that involve binary collisions have been studied by Wennberg and his 
students \cite{WW}. 
 
We now turn to a careful specification of the model we consider. 

\subsection{The Gaussian thermostatted particle system} 

The microscopic configuration of our system is given by 
$\X=(\q_1,\ldots,\q_N,\v_1,\ldots,\v_N)=(\Q,\V)$, with $\q_i\in\mathbb{T}^d$, 
$\v_i\in\R^d$. The time evolution of the system between collisions is given by:
\begin{equation}\label{motion}
\begin{cases}
\dot\q_i=\v_i\ ,&\qquad\qquad\q_i\in\mathbb{T}^d\ ,\\
\dot\v_i=\E-\frac{\displaystyle\E\cdot\j(\V)}{\displaystyle 
u(\V)}\v_i:=\F_{i}(\V)\ , 
&\qquad\qquad\v_i\in\mathbb{R}^d\ ,
\end{cases}
\end{equation}
where
\begin{equation}
\j(\V)=\frac1N\sum_{i=1}^N\v_i\ ,\qquad \qquad
u(\V)=\frac1N\sum_{i=1}^N\v_i^2\ .
\end{equation}
The term $\frac{\E\cdot\j}{u}\v_i$ in \eqref{motion} represents the Gaussian 
thermostat. It ensures that the kinetic energy per particle $u(\V)$ is constant 
in time. Each particle also undergoes collisions with ``virtual'' scatterers at 
random times according to a Poisson process with unit intensity which is 
independent of its velocity or position. A collision changes the direction but 
not the magnitude of the particle velocity. When a ``virtual'' collision takes 
place with incoming velocity $\v$, a unit vector $\hat\n$ is randomly selected 
according to a distribution $k(\hat\v\cdot\hat\n)$, where $\hat\v=\v/|\v|$, and 
the velocity of the particle is changed from $\v$ to $\v'$, where
\[
\v'=\v-2(\v\cdot\hat \n)\hat\n\ .
\]
For $d=1$ this means that $\v'=-\v$. 

\begin{remark}\label{equlik}
Observe that since $\hat\n$ and $-\hat\n$ both yield the same outgoing velocity,
we may freely suppose that $k$ is an even function on $[-1,1]$. Since
$\hat\v\cdot\hat\v'  = 1 -2(\hat\v\cdot\hat\n)^2$, the two collision $\v \to
\v'$ and $\w \to \w'$ are equally likely if
\[
\hat \v \cdot \hat \v' = \hat\w \cdot \hat \w'\ .
\]
Moreover, for any two vectors $\v$ and
$\v'$ satisfying $|\v| = |\v'|$ there is a collision that takes $\v$ into $\v'$:
Simply take $\hat\n = (\v - \v')/(|\v - \v'|)$.
\end{remark}

In the absence of an electric field, $\E=0$, each particle would move 
independently keeping its kinetic energy $\v_i^2$ constant while the collisions 
would cause the velocity $\v_i$ to become uniformly distributed on the $(d-1)$ 
dimensional sphere of radius $|\v_i|$. The force exerted by the Gaussian 
thermostat on the $i$-th particle, $(\E\cdot\j)\v_i/u$ induces, through $\j$, a 
mean-field type of interaction between the particles. When a particle speeds up 
(slows down) due to the interaction with the field some other particles have to 
slow down (speed up).

The probability distribution $W(\Q,\V,t)$ for the position and velocity of the 
particles will satisfy the master equation
\begin{equation}\label{master}
 \frac{\partial W(\Q,\V,t)}{\partial t}+\sum_{i=1}^N
\v_i\frac{\partial}{\partial \q_i}W(\Q,\V,t)+ \sum_{i=1}^N
\frac{\partial}{\partial
\v_i}\left[\left(\E-\frac{\E\cdot\j(\V)}{u(\V)} \v_i\right)W(\Q,\V,t)\right]=\cC
W(\Q,\V,t) \ ,
\end{equation}
where
\begin{equation}
 \cC W(\V,t)=\sum_{j=1}^N\int\left [
W(\V'_j,t)-W(\V,t)\right]k(\hat\v_j'\cdot\hat \n)\,d\hat\n\ ,
\end{equation}
$\V_j'=(\v_1,\ldots,\v_{j-1},\v'_j,\v_{j+1},\ldots,\v_N)$ and the integration is 
over the angles $\n$ such that $\v_j'\to\v_j$. Eq.\eqref{master} is to be solved 
subject to some initial condition $W(\Q,\V,0)=W_0(\Q,\V)$. It is easy to see 
that the projection of $W(\Q,\V,t)$ on any energy surface $u(\V)=e$ will evolve 
autonomously. This evolution will satisfy the D\"oblin condition \cite{Doe} for 
$\E\not=0$. Hence for any finite $N$ the effect of the stochastic collisions 
combined with the thermostated force is to make the $W(\Q,\V,t)$ approach a 
unique stationary state on each energy surface as $t\to\infty$. In fact this 
approach will be exponential for any finite $N$. However in this paper we rather 
focus on the time evolution of the one particle marginal distribution 
$f(\q,\v,t)$ for fixed $t$ as $N\to\infty$.

Now it seems reasonable to believe that in this limit one can use the law of 
large numbers to replace $\j(\V)$ in \eqref{master} by its expectation value 
with respect to $W$. When this is true, starting from a product measure
\begin{equation}\label{prod}
 W_0(\Q,\V)=\prod_{i=1}^N f_0(\q_i,\v_i)\ ,
\end{equation}
the system will stay in a product measure
\[
 W(\Q,\V,t)=\prod_{i=1}^N f(\q_i,\v_i,t)\ ,
\]
and $f(\q,\v,t)$ will satisfy an autonomous non-linear Vlasov-Boltzmann
equation (VBE)
\begin{equation}\label{BE}
 \frac{\partial f(\q,\v,t)}{\partial t}+ \v\frac{\partial}{\partial
\q}f(\q,\v,t) +\frac{\partial}{\partial \v}
\left[\left(\E-\frac{\E\cdot\tilde\j(t)}{\tilde u} \v\right)f(\q,\v,t)\right]=
\int
k(\hat\v'\cdot\hat\n)[f(\q,\v',t)-f(\q,\v,t)]d\hat\n
\end{equation}
where $\tilde \j(t)$ and $\tilde u$ are given by
\begin{equation}\label{cur}
 \tilde \j(t)=\int_{\R^d\times\mathbb{T}^d}\v f(\q,\v,t)\,{\rm d}\q{\rm d}\v\ ,
\qquad \qquad\tilde u=\int_{\R^d\times\mathbb{T}^d}|\v|^2 f(\q,\v,t)\,{\rm
d}\q{\rm d}\v\ .
\end{equation}
It is easy to check that $d\tilde u/dt=0$ and so $\tilde u(t)=\tilde u(0)$
independent of $t$. 

An important feature of the VBE  (\ref{BE}) is that the current $ \tilde \j(t)$ 
given by (\ref{cur}) satisfies an autonomous equation: \begin{lm} \label{curau} 
Let $f(\q,\v,t) $ be a solution of the VBE (\ref{BE}), and let the corresponding 
current $ \tilde \j(t)$ be defined by (\ref{cur}). Then $ \tilde \j(t)$ 
satisfies the equation 
\begin{equation}\label{cureq}
 \frac{d\tilde\j}{dt}=\E-\frac{\E\cdot\tilde\j}{\tilde 
u}\tilde\j-\rho_k\tilde \j\ ,
\end{equation}
where, for $d\geq 2$, 
\begin{equation}\label{lambdakdef}
\rho_k = 2|S^{d-2}| \int_0^\pi  k(\cos\theta)\cos^2\theta \sin^{d-2}\theta 
d\theta\ ,
\end{equation}
in which  $|S^{d-2}|$ denotes the area of $S^{d-2}$ in $\R^{d-1}$. In 
particular, for $d=2$, $S^{d-2} = \{-1,1\}$, and $|S^{d-2}|  =2$. For $d=1$, we 
have $\rho_k = 1$. 
\end{lm}

\begin{proof} We treat the case $d\geq 2$. By the definition of  $ \tilde \j(t)$ 
it suffices to show that 
\begin{equation}\label{must}
\int_{\R^d}   \v\left[\int_{S^{d-1}}
 k(\hat\v'\cdot\hat\n)f(\q,\v',t)d\hat\n\right]  d\v  = \tilde \j - 
\rho_k\tilde \j\ \ .
\end{equation}
Since $d\hat\n d\v = d\hat\n d\v'$, the left hand side of (\ref{must}) equals
\[
\int_{\R^d} f(\q,\v',t)\left[   \int_{S^{d-1}}
 \v k(\hat\v'\cdot\hat\n)d\hat\n\right] d\v'\ .
\]
Note  that 
\begin{eqnarray} 
 \int_{S^{d-1}} \v k(\hat\v'\cdot\hat\n)d\hat\n &= &  \int_{S^{d-1}} (\v'  - 
2[\v'\cdot \hat\n]\hat \n )k(\hat\v'\cdot\hat\n)d\hat\n\nonumber\\
&=&  \v'  - 2  \int_{S^{d-1}} (\v'\cdot \hat\n)
k(\hat\v'\cdot\hat\n)\hat \n d\hat\n\ .\nonumber  
\end{eqnarray}
To evaluate this last integral, introduce coordinates in which $\v' = 
(0,\dots,0,|\v'|)$, and parameterize $S^{d-1}$ by $\hat\n = 
(\boldsymbol{\omega}\sin\theta,\cos\theta)$ where $\boldsymbol{\omega}\in 
S^{d-2}$ and  $\theta \in [0,\pi]$.  Then $d\hat \n = \sin^{n-2}\theta 
d\boldsymbol{\omega}d\theta $, and since 
$\int_{S^{d-2}}\boldsymbol{\omega}d\boldsymbol{\omega} = 0$, 
\begin{eqnarray}
2  \int_{S^{d-1}} [\v'\cdot \hat\n]\hat \n k(\hat\v'\cdot\hat\n)d\hat\n  &=&
\left[2|S^{d-2}| \int_0^\pi  k(\cos\theta)\cos^2\theta \sin^{d-2}\theta d\theta\right] (0,\dots,0,|\v'|)\nonumber\\
&=& \rho_k \v'\ .\nonumber
\end{eqnarray}
Combining calculations yields (\ref{must}), and completes the  proof.
\end{proof}

The initial value problem for the current equation (\ref{cureq}) is readily 
solved: Write $\tilde \j = \tilde \j_\| + \tilde \j_\perp$ where $\tilde \j_\| 
= |\E|^{-2}(\tilde \j\cdot \E)\E$ is the component of $\tilde \j$ parallel to 
$\E$. Then defining $y(t) := |\E|^{-1}(\tilde \j \cdot \E)$, 
\begin{equation}\label{cureq2}
 \frac{dy}{dt}=|\E|-\frac{|\E|}{\tilde u}y^2-\rho_ky \qquad{\rm and}\qquad 
  \frac{d\tilde\j_\perp}{dt}=-\frac{|\E| y}{\tilde u}\j_\perp-\rho_k\tilde 
\j_\perp\ .
\end{equation}
The equation for $\tilde \j_\perp$  is easily solved once the equation for $y$ 
is solved. Rescaling $y\to y|E|=z$, and writing the equation for $z$ in the form 
$z' = v(z)$, so that $v(z) = |\E|^2 - z^2/\tilde u - \rho_k z$, we see that 
$v(z) = 0$ at 
\[
z_\pm =  \frac{\tilde u}{2}\left[   -\rho_k \pm 
\sqrt{\frac{4|\E|^2}{\tilde u} + \rho_k^2}  \right] \ ,
\]
and hence 
\[
y_\pm =  \frac{\tilde u}{2|E|}\left[   -\rho_k \pm 
\sqrt{\frac{4|\E|^2}{\tilde u} + \rho_k^2}  \right] \ .
\]
 Since $\rho_k>0$, $y_{- }< -\sqrt{\tilde u}$, and by the definitions of 
$\tilde \j$  and $y$, $y^2 (0) < \tilde u$.  Hence the equation for $y$ has a 
unique solution with $\lim_{t\to\infty}y(t) = y_+$. The equation is  solvable by 
separating variables, but we shall not need the explicit form here, though we 
use the uniqueness below. 

The fact that $\tilde \j$ solves an autonomous equation that has a unique 
solution allows us to proves existence and uniqueness for the spatially 
homogeneous for of the VBE (\ref{BE}); see the appendix.
 
In the large $N$ limit, the mean current of the particle model will satisfy this 
same equation as a consequence of POC: It follows from \eqref{master} that
\[
  \frac{d\langle\j\rangle}{dt}=\E-\E\cdot\left\langle\frac{\j
\j}{u}\right\rangle-\rho_k \langle\j\rangle\ ,
\]
where 
\[
\langle \psi \rangle=\int \psi(\V) W(\Q,\V,t)\,d\V\ .
\]
Propagation of chaos then shows that, as expected, in the limit $N\to\infty$,
\[\left\langle
\frac{\j\j}{u}
\right\rangle=\frac{\tilde\j(t)}{u}
\tilde\j(t)=\frac{\langle\j\rangle}{\langle u\rangle}\langle\j\rangle\ .
\]
Our main result is contained in the following theorem. For its proof we will 
need to control the fourth moment of $f(\q,\v,t)$. We thus set
\[
a(t) = \int_{\R^d\times\mathbb{T}^d}(|\v|^2-\tilde u)^2f(\q,\v,t){\rm d}\q{\rm 
d}\v\ .
\]
We also remind the reader that a function $\varphi$ on $\mathbb{R}^d$ is 
$1$-Lipschitz if for any $\mathbf{x},\mathbf{y}\in\mathbb{R}^d$ we have 
$|\varphi(\mathbf{x})- \varphi(\mathbf{y})|\leq |\mathbf{x}-\mathbf{x}|$.

\bigskip

\begin{thm}\label{chaos} Let $f_0(\q,\v)$ be a probability density on $\R$ with
\begin{equation}\label{cond}
\tilde u= \int_{\R^d\times\mathbb{T}^d} |\v|^2f_0(\q,\v){\rm d}\q{\rm d}\v
<\infty\qquad{\rm
and}\qquad a(0) = \int_{\R^d\times\mathbb{T}^d}
(|\v|^2-\tilde u)^2f_0(\q,\v){\rm d}\q{\rm d}\v  < \infty\ .
\end{equation}
Let $ W(\Q,\V,t)$ be the solution of \eqref{master} with  initial data
$W_0(\Q,\V)$ given by (\ref{prod}).  Then for all $1$-Lipschitz function
$\varphi$ on $\left(\R^{d}\times\mathbb{T}^{d}\right)^2$, and all $t>0$,
\begin{align}\label{chaoseq}
\lim_{N\to\infty} \int_{\R^{dN}\times\mathbb{T}^{dN}}&\varphi(\q_1,\v_1, \q_2,
\v_2)W(\Q,\V,t)
=\crcr
&\int_{\R^{2d}\times\mathbb{T}^{2d}}\varphi(\q_1,\v_1,\q_2,\v_2)f(\q_1,\v_1,
t)f(\q_2,\v_2 , t) { \rm d}\q_1{ \rm d}\v_1{\rm d}\q_2{\rm d}\v_2\ ,
\end{align}
where $f(\q,\v,t)$ is the solution of (\ref{BE}) with initial data $f_0(\q,\v)$.
\end{thm}

Because the interaction between the particles does not depend at all on their 
locations, the main work in proving this theorem goes into estimates for 
treating the spatially homogeneous case. Once the spatially homogeneous version 
is proved, it will be easy to treat the general case. 

We now begin the proof, starting with the introduction of an auxiliary process 
that propagates independence. 

\section{The B process.}

By integrating \eqref{master} on the position variables $\Q$ we see that
the marginal of $W$ on the velocities
\[
 W(\V,t)=\int_{\mathbb{T}^{dN}}W(\Q,\V,t)\dd\Q\ ,
\]
satisfies the autonomous equation
\begin{equation}\label{mastera}
\frac{\partial W(\V,t)}{\partial t}+\sum_{i=1}^n
\frac{\partial}{\partial
\v_i}\left[\left(\E-\frac{\E\cdot\j}{u} \v_i\right)W(\V,t)\right]=\cC W(\V,t)\ .
\end{equation}
This implies that, starting from an initial data $f_0(\v)$ in
\eqref{prod} independent from $\q$ the solution of \eqref{master} will
remain independent of $\Q$ and will satisfy \eqref{mastera}. We will thus
ignore the position and study propagation of chaos for \eqref{mastera}. 
Theorem \ref{chaos} will be an easy corollary of the analysis in this section,
as we show at the end of Section 4.

We start the system with an initial state
\begin{equation}\label{proda}
 W_0(\V)=\prod_{i=1}^N f_0(\v)\ .
\end{equation}
Given $f_0(\v)$ we can solve the BE equation \eqref{BE} and find $f(\v,t)$, see 
Appendix A.

To compare the solution of the Boltzmann equation to the solution of the master 
equation \eqref{mastera} we introduce a new stochastic process. Consider a 
system of $N$ particles with the following dynamics between collisions
\begin{equation}\label{motionq}
\left\{
\begin{array}{l l}
\dot\q_i=\v_i\ , &\qquad\qquad\q_i\in\mathbb{T}^d\ ,\crcr
\dot\v_i=\E-\frac{\E\cdot\tilde\j(t)}{\tilde u}\v_i=\widetilde \F_{i}(\V)\ ,
&\qquad\qquad\v_i\in\mathbb{R}^d\ ,
\end{array}
\right.
\end{equation}
where $\tilde\j(t)$ is obtained from the solution of \eqref{cureq}. The
particles also undergo ``virtual'' collision with the same distribution
described after \eqref{motion}. It is easy to see that the master
equation associated to this process is
\begin{equation}\label{masterq}
 \frac{\partial \widetilde W(\V,t)}{\partial t}+\sum_{i=1}^n
\frac{\partial}{\partial
\v_i}\left[\left(\E-\frac{\E\cdot\tilde\j(t)}{\tilde u} \v_i\right)\widetilde
W(\V,t)\right]=\cC \widetilde
W(\V,t)\ ,
\end{equation}
where we have assumed that $\widetilde W(\V,0)$ is independent of $\Q$.
Clearly, under this dynamics each particle evolves independently so
that if the initial condition for this system is of the form $\widetilde
W_0(\V)$
of \eqref{proda} then $\widetilde W(\V,t)$ will also have that form,
\begin{equation}\label{prodat}
 \widetilde W(\V,t)=\prod_i f(\v_i,t)\ ,
\end{equation}
where $f(\v,t)$ solves the spatially homogeneous BE \eqref{BE}.

We shall call the system described by the dynamics \eqref{motionq} the 
``Boltzmann'' or B-system and the original system described by the dynamics 
\eqref{motion} the A-system. We will prove that, for fixed $t$ and $N\to\infty$ 
every trajectory of the A-system with the dynamics \eqref{motion} can be made 
close (in an appropriate sense) to that of the B-system with dynamics 
\eqref{motionq}. It will then follow that, starting the B-system and A-system with the 
same product measure,  the A-system pair correlation function will be close to 
that of the B-system, and since the B-dynamics propagates independence, the 
latter will be a product. In this way we shall see that \eqref{BE} is satisfied.

\subsection{Comparison of the A and B processes.}

To compare the path of the two processes, we first observe that a collision is
specified by its time $t$, the index $i$ of the particle that collides, and the
unit vector $\hat\n$ that specifies the post-collisional velocity. We will call
$\omega=\{(t_k,i_k,\hat\n_k)\}_{k=0}^n$, with $s<t_k<t_{k+1}<t$, a collision
history in $(s,t)$ with $n$ collisions. 

Let $\Psi_{t-s}(\V_s)$ be the flow generated by the autonomous dynamics 
\eqref{motion} and $\Phi_{s,t}(\V_s)$ the flow generated by the non autonomous 
dynamics \eqref{motionq} both starting at $\V_s\in \R^{dN}$ at time $s$. Given a 
collision history $\omega$ in $(s,t)$ with $n$ collisions the A-process  
starting from $\V_s$ at time $s$, has the path
\begin{equation}
\V(t,\V_s,\omega) = \Psi_{t-t_n}\circ
R_n\circ \Psi_{t_n-t_{n-1}}\circ R_{n-1}\cdots
R_1\Psi_{t_1-s}(\V_s)\ ,
\end{equation}
while the B-process has the path
\begin{equation}
\widetilde \V(t,\V_s,\omega) =  \Phi_{t_n,t}\circ
\widetilde R_n\circ \Phi_{t_{n-1},t_n}\circ
\widetilde R_{n-1}\cdots
\widetilde R_1\Phi_{s,t_1}(\V_s)\ ,
\end{equation}
where $R_n$, respectively $\widetilde R_n$, is the updating of the velocity 
$\v_{i_n}$ of the $i_n$-th particle generated by the $n$-th collision in the A  
(resp. B) -dynamics. To simplify notation, here and in what follows, we will use 
$\V_s=\V(s,\V,\omega)$ and $\widetilde \V_s=\widetilde\V(s,\V,\omega)$ for the A 
and B processes.

We want to compare the trajectories of the two processes. We first observe that, 
given two incoming velocities $\v$ and $\w$ and a collision vector $\hat\n$ for 
$\v$ we can always find a corresponding collision vector $\hat\m$ for $\w$ such 
that:
\begin{itemize}
 \item The probability of selecting $\hat\n$ for an incoming velocity $\v$
is equal to the probability of selecting $\hat\m$ for an incoming velocity
$\w$. By Remark~\ref{equlik}, this is equivalent to having 
\begin{equation}\label{elik}
\hat\v\cdot\hat\v'=\hat\w\cdot\hat\w' \ .
\end{equation}

 \item The distance between the two outgoing velocities is equal to the distance
between the incoming velocities, i.e. $|\w'-\v'|=|\w-\v|$.  Since $|\v|= |\v'|$
and $|\w| = |\w'|$, this is equivalent to
\begin{equation}\label{edis}
\hat\v\cdot\hat\w=\hat\v'\cdot\hat\w' \ .
\end{equation}
\end{itemize}

Our claim amounts to:

\begin{prop}
Given any three unit vectors $\hat \v$, $\hat \v'$ and $\hat \w$, we can find a
fourth unit vector $\hat \w'$ such that both (\ref{elik}) and (\ref{edis}) are
satisfied. 
\end{prop}

\begin{proof}
For $d=1$, a solution is given by $\hat\w' =(\hat\v\cdot \hat\v')\hat\w$, i.e. 
$\hat
\w'=-\hat \w$. For $d=2$, let $R$ be any planar rotation such that $\hat \v' =
R\hat \v$, and define $\hat \w' = R\hat \w$. Then (\ref{elik}) and (\ref{edis})
reduce to $\hat \v\cdot R \hat \v =   \hat \w\cdot R \hat \w$ and $\hat \v\cdot
\hat \w =  R\hat \v\cdot  R\hat \w$, so that $\hat \w' = R\hat \w$ is a
solution. 

For $d\geq 3$, we reduce to the planar case as follows:  Let $P$ be the
orthogonal projection onto the plane spanned by $\v'$ and $\w$. Then $\w'$ is a
solution of (\ref{elik}) and (\ref{edis}) if and only if
\begin{equation}\label{edl}
P\hat\v\cdot\hat\v'=\hat\w\cdot P\hat\w' \qquad{\rm and}\qquad   P\hat\v\cdot\hat\w=\hat\v'\cdot P\hat\w' \ .
\end{equation}
Let us seek a solution with $|P\hat \w'| = |P\hat\v|$. If $|P\hat\v| =0$, we
simply choose $\hat\w'$ so that  $P\hat\w' =0$. Otherwise, divide both
\eqref{edl} by $|P\hat\w'|
= |P\hat\v|$, and we have reduced to the planar case: If $R$ is a rotation in
the plane spanned by $\hat\v'$ and $\hat\w$ such that $R\hat\v =
P\hat\v'/|P\hat\v'|$, then we define $P\hat\w' = |P\hat \v|R\hat\w$, and
finally 
\[
\hat\w' =  |P\hat\v|R\hat\w + \sqrt{1 - |P\hat\v|^2}\hat{\mathbf{z}}\ ,
\]
 where $\hat{\mathbf{z}}$ is any unit vector orthogonal to $\hat\v'$ and
$\hat\w'$. 
 \end{proof}

From the proposition it follows that to every collision history $\omega$ for the 
A process we can associate a collision history $p(\omega)$ for the B process 
such that for every collision time $t_k$ we have that
\begin{equation}\label{iso}
 \|\V(t_k^-,\V_0,\omega)-\widetilde \V(t_k^-,\V_0,p(\omega))\|_N=
\|\V(t_k^+,\V_0,\omega)-\widetilde \V(t_k^+,\V_0,p(\omega))\|_N\ ,
\end{equation}
where $\V(t_k^-,\V_0,\omega)$, resp. $\V(t_k^+,\V_0,\omega)$, is the path of the
A-process just before, resp. just after, the $k$-th collision, and similarly
for the B-process, and given a vector $\V\in \R^{dN}$ we set
\[
\|\V\|_N=\left(\frac{1}{N}\sum_{i=1}^N
\v_i^2\right)^{1/2}=\sqrt{\frac{1}{N}(\V\cdot \V)}\ .
\]
This means that we can see the A and B processes as taking place on the same
probability sample space defined by the initial distribution and the collision
histories $\omega$. We will call $\mathbb{P}$ the probability measure on this
space.

The result of this comparison is summarized in the following:

\begin{thm}[Pathwise Comparison of the Two Processes]\label{pathcomp}  
 Let $W_0(\V)$ satisfy \eqref{proda} with $f_0(\v)$ satisfying \eqref{cond}.
Then for all $\epsilon>0$,
\begin{multline*}
\mathbb{P}\left\{  \|\V(t,\V_0,\omega)-\widetilde \V(t,\V_0,\omega)\|_N >
\epsilon\right\} \leq \\
\frac {1}{\epsilon} \exp\left(4t\frac{\sqrt{2}}{\sqrt{\tilde u}}\right)
\frac{E t}{\sqrt{N}} \left[ 1 +  e^{t3E/\sqrt{\tilde u}}\frac{(a(0) + 
3\tilde u^2)^{1/2}}{\tilde u}\right]
+  \frac{4te^{t6E/\sqrt{\tilde u}}(a(0) + 3\tilde u^2)}{\Delta N\tilde u^2}\ ,
\end{multline*}
where $E=|\E|$ and $\Delta>0$ is defined in \eqref{comp3} below.
\end{thm}

\section{Proof of Theorem \ref{pathcomp}.}

To prove this theorem we need to find an expression for
$\V(t,\V_0,\omega)-\widetilde\V(t,\V_0,\omega)$.
First observe that between collisions we have
\[
 \Phi_{0,t}(\V)-\Psi_t(\V)=\Psi_{t-s}(\Phi_{0,s}(\V))\Bigr|_{s=0}^t\ .
\]
Differentiating the above expression with respect to $s$ and calling
$D\Psi_{t}$ the differential of the flow $\Psi_t$, i.e.
\[
 \bigl[D\Psi_t(\V)\bigr]_{i,j}=\frac{\partial[\Psi_{t}(\V)]_i}{\partial \v_j},
\]
we get
\begin{equation}\label{trotter}
\Phi_{0,t}(\V)-\Psi_t(\V)=\int_0^t
D\Psi_{t-s}(\Phi_{0,s}(\V))(\F(\Phi_{0,s}(\V))-
\widetilde\F(\Phi_{0,s}(\V)))\dd s\ .
\end{equation}
If the two flows start from different initial condition we get
\begin{eqnarray}
\Phi_{0,t}(\V)-\Psi_t(\widetilde\V)&=&
\int_{t_n}^t
D\Psi_{t-s}(\Phi_{0,s}(\V))(\F(\Phi_{0,s}(\V))-
\widetilde\F(\Phi_{0,s}(\V)))\dd s-
\nonumber\\
&&\int_0^1 D\Psi_{t}(\sigma
\V+(1-\sigma)\widetilde\V)
)(\V -\widetilde\V)\dd \sigma\ .\label{trotter1}
\end{eqnarray}
where $\sigma\V+(1-\sigma)\widetilde\V$ is the segment from $V$ to $\widetilde
V$.

The main idea of the proof is to iterate \eqref{trotter1} collision
after collision, using the fact that the distance between the A and B
trajectory is preserved by the collisions.

To control \eqref{trotter1} we need first to show that
\begin{equation}\label{quant1}
\|D\Psi_{t}(\V)\|=\sup_{\|\W\|_N=1}\|D\Psi_{t}
(\V)\W\|_N \ ,
\end{equation}
can be bounded uniformly in $N$. We will obtain such an estimate in Section 
\ref{Lyap} but our estimate will depend on $u(\V)^{-\frac 12}$. Since the B -process does not preserve the total energy of the particles we will need to show 
that $u(\widetilde \V_s)^{\frac 12}$ is, with a large probability uniformly in 
$N$, bounded away from 0. This is the content of the technical Lemma in Section 
\ref{lemma}.

Having bounded the differential of $\Psi_t$ we will need to show that
\begin{equation}\label{quant2}
\int_0^t \|\F(\widetilde \V_s)-\widetilde\F(\widetilde\V_s)\|_N {\rm d}s
\end{equation}
is small uniformly in $N$. This is the content of Section \ref{quant2} and it is
based on the Law of Large Numbers.

Thus we will have that the distance between the corresponding paths of the two
processes at time $t$ can be estimated in term of the norm \eqref{quant1} of the
differential of $\Psi_{t}$, the integral \eqref{quant2} and the distance just
after the last collision at $t^+_n$. Using \eqref{iso} we can reduce the problem
to an estimate of the distance before the last collision. Finally we can get a
full estimate iterating the above procedure over all the collisions.

\subsection{A Technical Lemma}\label{lemma}

While the propagation of independence is an important advantage of the master
equation (\ref{masterq}), one disadvantage is that, while the A-evolution
(\ref{motion}) conserves the energy $u(\V)$, the B evolution does not: it
only does so in the average.

In what follows we will  need a bound on $\sup_{0 \leq s \leq
t}\left(u(\widetilde \V(s))\right)^{-1/2}$ showing that the event that this
quantity is large has a small probability. 

\begin{lm}\label{sr}
Let $\widetilde \V_t$ be the B process with initial data given by (\ref{proda})
with $f_0(\v)$ satisfying \eqref{cond}. Define

\begin{equation}\label{comp3}
\Delta :=  \frac{\sqrt{\tilde u}}{E} \log \left(\frac{2+2\sqrt{2}}{
1+2\sqrt{2}}\right)\ ,
\end{equation}
with $\tilde u$ given by \eqref{cond} and for any given $t>0$ let $m(t)$ denote 
the least integer $m$ such that
$m\Delta \geq t$. Then

\begin{equation}\label{stabbnd}
 \mathbb{P}\left\{ \sup_{0 \leq s \leq t}\left(u(\widetilde \V(s))\right)^{-1/2}
\geq  \frac{2\sqrt{2}}{\sqrt{\tilde u}}
 \right\} \leq  m(t)\frac{4e^{t6E/\sqrt{\tilde u}}(a(0) + 3\tilde u^2)}{N\tilde
u^2}\ .
\end{equation}

\end{lm}
 
\begin{proof} 
Observe that the collisions do not affect $u$. Thus $u(\widetilde
\V(t))$ is path-wise differentiable, and 
\[
\frac{{\rm d}}{{\rm d}t}u(\widetilde \V_t)  =  2E\j(\widetilde \V_t)  -
2E\frac{\tilde \j(t)}{\tilde u}u(\widetilde \V_t)\ .
\]
By the Schwarz inequality, $|\j(\widetilde \V)| \leq \left(u(\widetilde
\V)\right)^{1/2}$ and $\tilde \j(t) \leq \sqrt{\tilde u}$, Therefore,

\begin{equation}\label{compA}
 \frac{{\rm d}}{{\rm d}t}\left(u(\widetilde \V_s)\right)^{-1/2}  \leq
E\left(u(\widetilde \V_s)\right)^{-1} +  \frac{E}{\sqrt{\tilde u}}
 \left(u(\widetilde \V_s)\right)^{-1/2}\ .
 \end{equation}
The differential equation
  
\begin{equation}\label{compB}
 \dot x(t) = Ex^2(t)+\frac{E}{\sqrt{\tilde u}}x(t)\quad, \qquad  x(t_0) = x_0
\end{equation}
is solved by
\[
x(t) = \left[e^{-(E/\sqrt{\tilde u})(t-t_0)} \frac{1}{x_0} -
\sqrt{\tilde u}\left(1 -  e^{-(E/\sqrt{\tilde u})(t-t_0)}\right)\right]^{-1}\ .
\]
Let $t_1$ denote the time at which $x(t_1) = 2x_0$:
\[
e^{-(E/\sqrt{\tilde u})(t_1-t_0)}   = \frac{1+2\sqrt{\tilde u}x_0}{2 +
2\sqrt{\tilde u}x_0}\ .
\]
Let us choose $x_0 = \sqrt{2/\tilde u}$. Then
\begin{equation}\label{comp4}
t_1 - t_0 =  \frac{\sqrt{\tilde u}}{E} \log \left(\frac{2+2\sqrt{2}}{
1+2\sqrt{2}}\right)=\Delta\ .
\end{equation}
Note that the length of the interval $[t_0,t_1]$ is independent of $t_0$.  

It now follows from (\ref{compA}) and what we have proven about the differential 
equation  (\ref{compB}) that  if $1/\sqrt{u(\widetilde \V_{t_0})}\leq 
\sqrt{2/\tilde u}$, then for all $t$ on the interval $[t_0,t_1]$,  
$1/\sqrt{u(\widetilde \V_t)} \leq 2\sqrt{2/\tilde u}$. Then, provided that
\[
\frac{1}{\sqrt{u(\widetilde \V_{j\Delta t})}} \leq \sqrt{\frac{2}{\tilde u}}
\qquad{\rm for\
all}\qquad 0\leq j
\leq m(t)\,,
\]
we have
\[
\frac{1}{\sqrt{u(\widetilde \V_s)}} \leq 2\sqrt{\frac2{\tilde u}} \qquad{\rm
for\ all}
\qquad 0 \leq s \leq
t\ .
\]
Moreover, for any given $t_0$, there is only a small (for large $N$) probability
that $1/\sqrt{u(\widetilde \V_{t_0})} \geq \sqrt{2/\tilde u}$. Indeed
since $ \mathbb{E}\left[ u(\widetilde \V_{t_0}) -\tilde u\right] ^2= a(t_0)/N$,
we have
\begin{equation}\label{comp5}
  \mathbb{P}\left\{  u(\widetilde \V_{t_0}) < \tilde u/2\right\} \leq
\frac{4a(t_0)}{N\tilde u^2}\ .
\end{equation}
It follows that
\[
\mathbb{P}\left\{  \sup_{0 \leq s \leq t} \left(u(\widetilde \V_s) \right)\geq
\frac{2\sqrt{2}}{\sqrt{\tilde u}}
\right\} \leq \sum_{j=0}^{m(t)}  \mathbb{P}\left\{  u(\widetilde \V_{j\Delta t})
<
\tilde u/2\right\}\ .
\]
Combining this estimate with (\ref{comp5}) and (\ref{varbound}) from the
following Lemma leads to (\ref{stabbnd}). 

\end{proof}

\begin{lm}\label{envarlm} 
Let $f_0(\v)$ be a probability density on $\R^d$ satisfying \eqref{cond}
then for all $t>0$, 
\begin{equation}\label{varbound}
a(t) \leq e^{t6E/\sqrt{\tilde u}}(a(0) + 3\tilde u^2)\ .
\end{equation}
\end{lm}

\begin{proof}
Using  \eqref{BE}, we compute
\[
\frac{{\rm d}}{{\rm d}t} a(t) = -4E \frac{\tilde \j(t)}{\tilde u} a(t) + 4E
\int_\R (\v^2- \tilde u)\v f(\v,t)\,\dd\v\ .
\]
By the Schwarz inequality, $|\tilde\j(t)| \leq \sqrt{\tilde u}$ and
\[
 \int_{\R^d} (\v^2- \tilde u)\v f(\v,t)\,\dd\v\leq \sqrt{a(t)\tilde u}\ ,
\]
so that
\[
\frac{{\rm d}}{{\rm d}t} a(t) \leq 4\frac{E}{\sqrt{\tilde u}} (a(t) +
\sqrt{a(t)}\tilde u)
\leq  \frac{E}{\sqrt{u}} (6a(t)  + 2 \tilde u^2)\ .
\]

\end{proof}

\subsection{Estimate on (\ref{quant2}).}

The reason the quantity in (\ref{quant2}) will be small for large $N$ with high
probability is that the B process propagates independence, so that
$\widetilde \V_s$ has independent components. Thus, the Law of Large
Numbers says that with high probability, $\j(\widetilde \V_s)$ is very close
to $\tilde\j(s)$ and thus $\F(\widetilde \V_s) $ will be very close to
$\widetilde \F(\widetilde \V_s)$. The following proposition makes this
precise. 

\begin{prop}[Closeness of the Two Forces for Large $N$]\label{close} 
Let $f_0(\v)$ be a probability density on $\R^d$ satisfying \eqref{cond}. Then,
\[
\mathbb{E}\left( \int_0^t  \|\F(\widetilde \V_s)) -  \widetilde
\F(\widetilde \V_s)\|_N{\rm d}s\right) \leq
\frac{E t}{\sqrt{N}} \left[ 1 +  e^{t3E/\sqrt{\tilde u}}\frac{(a(0) + 
3\tilde u^2)^{1/2}}{\tilde u}\right]\ .
\]
where $\mathbb{E}$ is the expectation with respect to $\mathbb{P}$.
\end{prop}

\begin{proof}
From \eqref{motion} and \eqref{motionq} we have
\[
(\F(\widetilde \V_s) - \widetilde \F(\widetilde \V_s))_i =
E\left(\frac{\j(\widetilde \V_s)}{u(\widetilde \V_s)} -
\frac{\tilde\j(s)}{\tilde u}\right)\widetilde
\v_{i,s}\ .
\]
Clearly
\[
\left(\frac{\j(\V)}{u(\V)} - \frac{\tilde\j}{\tilde u}\right) = \frac{\j(\V) -
\tilde\j}{\tilde u}
+\frac {\tilde u -u(\V)}{\tilde uu(\V)} j(\V)\ .
\]
so that by Schwarz inequality and the triangle inequality, we get
\[
\|\F(\widetilde \V_s) - \widetilde \F(\widetilde \V_s)\|_N
\leq \frac{E}{\widetilde u}|\j(\widetilde \V_s) - \tilde\j(s)| \sqrt{u(\widetilde
\V_s)}
+ \frac{E}{\widetilde u}|u(\widetilde \V_s) - \tilde u|
\ . 
\]
Taking the expectation and using the Schwarz inequality again
\begin{eqnarray}
\mathbb{E}\left( \int_0^t  \|\F(\widetilde \V_s) -
\widetilde \F(\widetilde \V_s)\|_N{\rm d}s\right)  &\leq &
\frac{E}{\sqrt{\tilde u}}\int_0^t \left(\mathbb{E}|\j(\widetilde \V_s) -
\tilde\j(s)|^2\right)^{1/2}\nonumber\\
&+& \frac{E}{\tilde u}\int_0^t  \left(\mathbb{E}|u(\widetilde \V_s) -
u|^2\right)^{1/2}\ .\nonumber
\end{eqnarray}
We recall that for the B-process, the $N$ particle distribution at time $t$
is $\widetilde W(\V,t)=\prod_{j=1}^N f(\v_j,t)$.
Thus the expectation of $\j(\widetilde\V) - \tilde\j = \frac{1}{N}\sum_{j=1}^N
(\hat\v_{j,s} - \tilde\j(s))$ is given by
\[
\mathbb{E}|J(\widetilde \V_s)) - \tilde\j(s)|^2 =\int_{\mathbb{R}^{Nd}}
\Bigl|\frac{1}{N}\sum_{j=1}^N
(\v_{j} - \tilde\j(s))\Bigr|^2\prod_{j=1}^N f(\v_j,t)=
\frac{1}{N}\int_{\R^d} |\v- \tilde\j(s)|^2f(\v,t)\dd \v =  \frac{1}{N}(\tilde u 
-
\tilde\j^2(s))\leq
\frac{\tilde u}{N}\ .
\]
Likewise,
\[
\mathbb{E}|u(\widetilde \V_s) - \tilde u|^2  = \frac{1}{N}\int_\R
|\v^2- \tilde u|^2f(\v,t)\dd \v = \frac{1}{N}a(s) \leq
\frac{1}{N}e^{t6E/\sqrt{\tilde u}}(a(0)
+ 3\tilde u^2)\ ,
\]
where the last inequality comes form Lemma~\ref{envarlm}. 
Altogether, we have
\[
\mathbb{E}\left( \int_0^t  \|\F(\widetilde \V(s,\V_0,\omega)) -  \widetilde
\F(\widetilde \V(s,\V_0,\omega))\|_N{\rm d}s\right)  \leq \frac{E t}{\sqrt{N}}
\left[ 1 +  e^{t3E/\sqrt{\tilde u}}(a(0) + 3\tilde u^2)^{1/2}/\tilde u\right]\ .
\]
\end{proof}

\subsection{The Lyapunov Exponent.}\label{Lyap}

The next proposition gives an estimate on
$\|D\Psi_{t}(\widetilde\V)\|$ defined in \eqref{quant1} in term of $u(\V)$.

\begin{prop}\label{lyapunov} Given $\V\in \mathbb{R}^{Nd}$ we have:
\[
\|D\Psi_{t}(\V)\| \leq e^{t \lambda(\V)}\ ,
\]
where
\[
 \lambda(\V)=\frac{4}{\sqrt {u(\V)}}\ .
\]
\end{prop}
 
\begin{proof}
Clearly,
\[
 \frac{{\rm d}} {{\rm d}t}  D\Psi_t(\V)=D \F(\Psi_t(\V))D\Psi_t(\V)\ ,
\]
with initial condition $D\Psi_0(\V)={\rm Id}$. 
Given a vector $\W\in \R^N$ we get
\[
 \frac{d}{dt}\|D\Psi_t(\V)\W\|^2_N=\frac{2}{N}(\frac{d}{dt}D\Psi_t(\V)\W\cdot
D\Psi_t(\V)\W)=
\frac{2}{N}(D \F(\Psi_t(\V))D\Psi_t(\V)\W\cdot D\Psi_t(\V)\W)\ .
\]
We thus get
\[
 \frac{d}{dt}\|D\Psi_t(\V)\|_N\leq \|D \F(\Psi_t(\V))\|\|D\Psi_t(\V)\|_N\ ,
\]
where $\|D \F(\Psi_t(\V))\|=\sup_{\|\W\|=1}|(D \F(\Psi_t(\V))\W\cdot \W)|$ with
$\|\W\|=\sqrt{(\W\cdot\W)}$. We have
\[
 \frac{\partial}{\partial
v_i}\F_j=-\frac{1}{u(\V)}\frac{v_j}{N}+\frac{\j(\V)}{u(\V)^2}\frac{2v_iv_j}{N}-
\frac{\j(\V)}{u(\V)}\delta_{i,j}\ ,
\]
so that
\[
 DF(\V)=-\frac{1}{u(\V)}\frac{\V\otimes{\bf
1}}{N}+\frac{\j(\V)}{u(\V)^2}\frac{2\V\otimes \V}{N}-\frac{j(\V)}{u(\V)}\,{\rm
Id}\ .
\]
Clearly we have
\[
 \|\V\otimes {\bf 1}\|=\sqrt N\|\V\|\ ,\qquad \|\V\otimes\V\|=\|\V\|^2\ ,\qquad \|{\rm
Id}\|=1 \ .
\]
Finally, form $|\j(\V)|\leq \sqrt{u(V)}$ and the above estimates we get
\[
 \frac{d}{dt}\|D\Psi_t(\V)\|\leq \frac{4}{\sqrt{u(\V)}}\|D\Psi_t(\V)\|\ .
\]
Since the A-dynamics preserves
$u(\V)$, this differential inequality may be ingrates to obtain the stated bound. 
\end{proof}

\subsection{Proof of Theorem \ref{pathcomp}}

We can now conclude the proof of theorem \ref{pathcomp}. Given a set
$T=\{t_1,\ldots,t_n\}$ of $n$ collision times $0<t_1<\ldots<t_n<t$ let
$\Omega_T$ be the event that there are exactly $n$ collision in $(0,t)$ and they
take place at the times $t_k$, $k=1,\ldots,n$. We now estimate the growth of
$\|\V(t,\V_0,\omega)-\widetilde \V(t,\V_0,\omega)\|_N$ along each sample path
$\omega$. We will ``peel off'' the collision, one at a time. 

First we get
\begin{eqnarray*}
\|\V(t,\V_0,\omega)-\widetilde
\V(t,\V_0,\omega)\|_N   &=&   \|\Psi_{t-t_n}\circ R_n(\V(t_n^- ,\V_0,\omega))-
\Phi_{t_n,t}\circ \widetilde R_n(\widetilde \V(t_n^-
,\V_0,\omega)\|_N\\
&\leq &   \|\Psi_{t-t_n}\circ R_n(\V(t_n^- ,\V_0,\omega))- \Psi_{t-t_n}\circ
\widetilde R_n(\widetilde \V(t_n^- ,\V_0,\omega)\|_N\\
&+ &   \|\Psi_{t-t_n}\circ \widetilde R_n(\widetilde \V(t_n^- ,\V_0,\omega))-
\Phi_{t_n,t}\circ \widetilde R_n(\widetilde \V(t_n^-
,\V_0,\omega)\|_N\ ,
\end{eqnarray*}
where $\V(t_n^- ,\V_0,\omega)$ to denote the configuration just
before the $n$th collision, and likewise for $\widetilde \V$, and
$R_n$ and $\widetilde{R}_n$ are the corresponding updating of the velocities at
the $n$-th collision. We deal with the last term first. Since there are no
collisions between $t_n$ and $t$, and since the configurations coincide at
$t_n$, Proposition \ref{lyapunov} gives
\begin{align}
\|\Psi_{t-t_n}\circ \widetilde R_n(\widetilde \V(t_n^- ,\V_0,\omega))-&
\Phi_{t_n,t}\circ \widetilde R_n(\widetilde \V(t_n^- ,\V_0,\omega)\|_N\leq\\
& e^{(t-t_n)\lambda(\widetilde V_{t_n^-})}\int_{t_n}^t  \|\F(\widetilde
\V(s,\V_0,\omega))-\widetilde\F(\widetilde
\V(s,\V_0,\omega))\|_N {\rm d}s\ .
\end{align}
Now,  connect  $R_n(V(t_n^-,\V_0,\omega))$ and   $\widetilde R_n(\widetilde
V(t_N^-,\V_0,\omega)$  by a straight line segment.  Again by
Proposition~\ref{lyapunov}, we get
\begin{align}
\|\Psi_{t-t_n}\circ R_n(\V(t_n^- ,\V_0,\omega))-&\Phi_{t_n,t}\circ
\widetilde R_n(\widetilde \V(t_n^- ,\V_0,\omega))\|_N \nonumber\\
\leq& e^{(t-t_n)\lambda_n} \|R_n(\V(t_n^- ,\V_0,\omega)- \widetilde
R_n(\widetilde \V(t_n^- ,\V_0,\omega)\|_N\nonumber\\
=& e^{(t-t_n)\lambda_n} \|\V(t_n^- ,\V_0,\omega)- \widetilde \V(t_n^-
,\V_0,\omega)\|_N\ ,
\end{align}
where $\lambda_n$ is the maximum of $4/\sqrt{u(\V)}$ along the line segment
joining  $R_n(V(t_n^-,\V_0,\omega))$ and   $\widetilde R_n(\widetilde
V(t_n^-,\V_0,\omega))$, and where we have used the distance preserving property
of $R_n$ and $\widetilde R_n$. Altogether, we have
\begin{eqnarray}
\|\V(t,\V_0,\omega)-\widetilde
\V(t,\V_0,\omega)\|_N   &\leq& e^{(t-t_n)\lambda_n} \left (\|\V(t_n^-
,\V_0,\omega)-\widetilde \V(t_n^-
,\V_0,\omega)\|_N\phantom{\int_{t_n}^t}\right.\nonumber\\ 
&+ & \left.   \int_{t_n}^t  \|\F(\widetilde
\V(s,\V_0,\omega))-\widetilde\F(\widetilde
\V(s,\V_0,\omega))\|_N {\rm d}s\right)\ .\nonumber
\end{eqnarray}

Applying the same procedure to estimate  $\|\V(t_n^- ,\V_0,\omega)-
\widetilde \V(t_n^- ,\V_0,\omega)\|_N$ in terms of $\|\V(t_{n-1}^-
,\V_0,\omega)- \widetilde \V(t_{n-1}^- ,\V_0,\omega)\|_N$, and so forth, we
obtain the estimate
\begin{equation}\label{mest}
\|\V(t,\V_0,\omega)-\widetilde
\V(t,\V_0,\omega)\|_N  \leq
 e^{t\max_{1\leq m \leq n}\lambda_m }
\int_{0}^{t}  \|\F(\widetilde
\V(s,\V_0,\omega))-\widetilde\F(\widetilde
\V(s,\V_0,\omega'))\|_N {\rm d}s\ .
\end{equation}  

Our next task is to control $\max_{1\leq m \leq n}\lambda_m$.
Let $\cA$ be the event that 
\[
\inf_{0 \leq s \leq t}\sqrt{u(\widetilde \V_s)} >
\frac{\sqrt{\tilde u}}{2\sqrt{2}}\ .
\]
Since
\[
|u(\V) - u(\widetilde \V)|  \leq \frac{1}{N}\sum_{j=1}^N |(\v_j -
\tilde \v_j)\cdot  (\v_j +
\tilde \v_j) \leq \|\V - \widetilde \V\|_N\left(\sqrt{ u(\V)} + \sqrt{
u(\widetilde \V)}\right)\ ,
\]
we get
\[
\sqrt{u(\V)} \geq   \sqrt{u(\widetilde \V)}  -  \|\V - \widetilde \V\|_N\ .
\]
Therefore, for any $0<\delta < \sqrt{\widetilde u/\sqrt{32}}$, if $u(\widetilde
\V) >
\widetilde u/8$ and $\widetilde \V$ is in the ball centered at $\V$ with radius 
$\delta$ in the $\|\cdot\|_N$ norm, we have   
\[
\frac{1}{\sqrt{u(\V)}} \leq \frac{1}{\sqrt{ \widetilde u/\sqrt{8}} -
\delta}\leq \frac{1}{\sqrt{ \widetilde u/\sqrt{32}}} .
\]
We can now define
\[
\delta_m :=  \|\V(t_m,\V_0,\omega)-\widetilde
\V(t_m,\V_0,\omega)\|_N\ ,
\]
and note $\delta_0 = 0$. The above estimates show that
\begin{equation}\label{fulc}
\delta_{m} \leq   \frac{\widetilde u}{4\sqrt{2}} \quad \Rightarrow \quad  \lambda_m \leq 
\frac{4\sqrt{2}}{\sqrt{ \widetilde u}} =: K\ .
\end{equation}
Let $m_\star$ be the least $m$ such that $\lambda_m >K$ if such an $m$
exists with $t_m < t$. We shall now show that this cannot happen in $\cA\cup 
\cB$ where $\cB$ is the event 
\begin{equation}\label{sec}
\int_{0}^{t}  \|\F(\widetilde
\V(s,\V_0,\omega))-\widetilde\F(\widetilde
\V(s,\V_0,\omega'))\|_N {\rm d}s < \frac{\widetilde ue^{-tK}}{4\sqrt{2}}\ .
\end{equation}
Rewriting the estimate \eqref{mest} for $m_\star$ we get
\[
 \delta_{m_\star}\leq e^{t_{m_\star}\max_{1\leq p \leq m_\star}\lambda_p }
\int_{0}^{t_{m_\star}}  \|\F(\widetilde
\V(s,\V_0,\omega))-\widetilde\F(\widetilde
\V(s,\V_0,\omega'))\|_N {\rm d}s\leq e^{t_{m_\star} K}\frac{\widetilde
ue^{-tK}}{4\sqrt{2}}\leq \frac{\widetilde u}{4\sqrt{2}}\ ,
\]
so that, if $t_{m_\star} < t$, we have, by \eqref{fulc}, $\lambda_{m^*} < K$.
Thus  no such $m_*$ exists. 

Let now $\cB'$ be the event
\[
\int_{0}^{t}  \|\F(\widetilde
\V(s,\V_0,\omega))-\widetilde\F(\widetilde
\V(s,\V_0,\omega'))\|_N {\rm d}s < \epsilon\ ,
\]
with $\epsilon$ sufficiently small, depending on $t$ and
$\widetilde u$, so that $\cB'\subset \cB$. On $\cA\cap
\cB$ we have that
\[
\|\V(t,\V_0,\omega)-\widetilde
\V(t,\V_0,\omega)\|_N  \leq e^{tK}\epsilon\ .
\]
Moreover we have
\begin{eqnarray*}
 \mathbb{P}(B'^c)&=&   \mathbb{P}\left(\int_0^t
\|\F(\widetilde \V(s,\V_0,\omega))-\widetilde\F(\widetilde 
\V(s,\V_0,\omega))\|_N >\epsilon\right)\nonumber\\
&\leq& \frac {1}{\epsilon}\int_0^t
\mathbb{E}\left(\|\F(\widetilde \V(s,\V_0,\omega))-\widetilde\F(\widetilde
\V(s,\V_0,\omega))\|_N \right){\rm d}s\ .\nonumber\\
\end{eqnarray*}
Finally Theorem \ref{pathcomp} follows using Lemma \ref{sr}, Proposition
\ref{close} and observing that $\mathbb{P}((\cA\cup \cB')^c)\leq
\mathbb{P}(\cA^c)+\mathbb{P}(\cB'^c)$.

\section{Propagation of chaos}

For each fixed $N$ and  $t>0$, we introduce the two empirical
distributions
\[
\mu_{N,t} = \frac{1}{N} \sum_{j=1}^N \delta_{\v_j(\omega,t)}  \qquad{\rm
and}\qquad 
\widetilde \mu_{N,t} = \frac{1}{N} \sum_{j=1}^N \delta_{\widetilde
\v_j(\omega,t)}\ ,
\]
where $\delta_{\v_j}=\delta(\v-\v_j)$. By the Law of  Large Numbers, we
have almost surely
\[
\lim_{N\to\infty}  \widetilde \mu_{N,t} = f(t,\v){\rm d}\v
\]
in distribution.  We now use this to complete the proof of Theorem \ref{chaos} 
in the spatially homogeneous case. 

By the permutation symmetry that follows from (\ref{prod}),
\[
\mathbb{E}\left(\varphi(\v_1(t),\v_2(t)\right)) =
\mathbb{E}\left(\frac{1}{N-1}\sum_{j=2}^{N}
(\varphi(\v_1(t),\v_j(t))\right)\ ,
\]
and moreover,
\begin{align*}
&\left|\frac{1}{N-1}  \sum_{j=2}^{N} \varphi(\v_1(t),\v_j(t))  - \int_{\R^d}
\varphi(\v_1(t),y)\dd \mu_{N,t}(y)\right|  =\\
&\qquad\qquad\left|\left(\frac{1}{N-1} -\frac{1}{N}\right) \sum_{j=2}^{N} 
\varphi(\v_1(t),\v_j(t))  - \frac{1}{N}
\varphi(\v_1(t),\v_1(t))\right|\leq 
\frac{2}{N}\|\varphi\|_\infty\ .
\end{align*}

Next, for every fixed ${\bf x}$,
\begin{eqnarray}
\left|\int_{\R^d} \varphi({\bf x},{\bf y})\dd \mu_{N,t}({\bf y}) - \int_{\R^d} 
\varphi({\bf x},{\bf y})\dd
\widetilde \mu_{N,t}({\bf x})\right|  &=&
\left| \frac{1}{N}\sum_{j=1}^N \varphi({\bf x},\v_j(t,\omega)) - 
 \frac{1}{N}\sum_{j=1}^N \varphi({\bf 
x},\widetilde{\v}_j(t,\omega))\right|\nonumber\\
 &\leq& \frac{1}{N} \sum_{j=1}^N |\v_j(t,\omega) - \widetilde{\v}_j(t,\omega)
|\nonumber\\
  &\leq&\left( \frac{1}{N} \sum_{j=1}^N |\v_j(t,\omega) - \widetilde
\v_j(t,\omega)|^2 \right)^{1/2}\ .\nonumber\\
  &=& \|\V(t,\V_0,\omega)-\widetilde \V(t,\V_0,\omega)\|_N\ .\nonumber
\end{eqnarray}
Therefore,
\begin{multline*}\left| \mathbb{E}(\varphi(\v_1(t),\v_2(t))) -
\mathbb{E}\left(\int_\R \varphi(\v_1(t)\ ,\ {\bf y})\dd
\widetilde\mu_{N,t}({\bf y})\right)\right| \leq\\ \epsilon + 
2\|\varphi\|_\infty
\mathbb{P}\left\{ 
  \|\V(t,\V_0,\omega)-\widetilde \V(t,\V_0,\omega')\|_N > \epsilon\right\}\ .
\end{multline*}
Choosing $\epsilon = N^{-1/4}$, we see that
\[
\lim_{N\to\infty} \left| \mathbb{E}(\varphi(\v_1(t),\v_2(t))) - 
\mathbb{E}\psi(\v_1(t))\right| = 0\ ,
\]
where $\psi({\bf x})$ is defined by
\[
\psi({\bf x}) := \int_{\R^d} \varphi({\bf x},{\bf y}) \dd \mu_{N,t}({\bf y})\ .
\]

Since $\varphi$ is $1$-Lipschitz on $\R^{2d}$, $\psi({\bf x})$ is $1$-Lipschitz on
$\R^d$. That is, for all $\v,\w\in \R^d$,
\[
|\psi(\v) - \psi(\w)|  \leq |\v-\w|\ .
\]
Moreover, $\|\psi\|_\infty \leq \|\varphi\|_\infty$. . 

Then, arguing just as above, 
\[
\mathbb{E}\psi(\v_1(t)) = \frac{1}{N} \sum_{j=1}^N \psi(\v_j(t))  =\int_{\R^d}
\psi(\v)\dd \mu_{N,t}(\v) \ ,
\]
and
\[
\left|
\int_\R \psi(\v)\dd \mu_{N,t}(\v)   -    \int_\R \psi(\w)\dd \widetilde
\mu_{N,t}(\w)
\right| \leq 
 \|\V(t,\V_0,\omega)-\widetilde \V(t,\V_0,\omega)\|_N\ .
\]
Therefore,
\[
\left|  \mathbb{E}\psi(\v_1(t))  - \int_\R \psi(\v)\dd \widetilde\mu_{N,t}(\v)
\right|  \leq  \epsilon + 2\|\varphi\|_\infty \mathbb{P}\left\{ 
  \|\V(t,\V_0,\omega)-\widetilde \V(t,\V_0,\omega')\|_N > \epsilon\right\}\ .
\]
Again, choosing $\epsilon = N^{-1/4}$, we see that
\[
\lim_{N\to\infty} \left|  \mathbb{E}\psi(\v_1(t))  - \int_{\R^d} \psi(\v)\dd
\widetilde\mu_{N,t}(\v) \right|= 0\ .
\]
Altogether, we now have

\[\lim_{N\to\infty}\left| \mathbb{E}(\varphi(\v_1(t),\v_2(t))  -   
\int_{\R^{2d}}
\varphi(\v,\w)\dd \widetilde\mu_{N,t}(\v) \dd \widetilde\mu_{N,t}(\w) \right| = 0\ .
\]
But by the Law of Large Numbers,
\begin{equation}\label{lln}
\lim_{N\to\infty} \int_{\R^{2d}} \varphi(\v,\w)\dd \widetilde\mu_{N,t}(\v) \dd
\widetilde\mu_{N,t}(\w)  = \int_{\R^{2d}}\varphi(\v,\w)f(t,\v)f(t,\w){\rm 
d}\v{\rm d}\w\ .
\end{equation}

This shows how our estimates give us propagation of chaos in the spatially 
homogeneous case. To prove Theorem~\ref{chaos}, we need only  explain how to 
use these same estimates to treat the spatial dependence. 

\begin{proof}[Proof of Theorem~\ref{chaos}]

It remains to take into account the spatial dependence. Note that 
$$\Q(t,\Q_0,\V_0,\omega) = \Q_0 + \int_0^t \V(s,\V_0,\omega){\rm d} s \qquad{\rm 
and}\qquad   \widetilde \Q(t,\Q_0,\V_0,\omega)  = \Q_0 + \int_0^t 
\widetilde\V(s,\V_0,\omega){\rm d} s\ .$$ It suffices to show that for each 
$\epsilon>0$,
\[
 \mathbb{P}\left\{ 
  \|  \Q(t,\Q_0,\V_0,\omega)  - \widetilde\Q(t,\Q_0,\V_0,\omega) \|_N   > \epsilon\right\}\ 
\]
goes to zero as $N$ goes to infinity, since in this case we may repeat the 
argument we have just made, but with test functions $\varphi$ on phase space. 

If we knew that  $\V(s,\V_0,\omega)$ is close to $\widetilde \V(s,\V_0,\omega)$ 
everywhere along most paths, we would conclude that $\Q(t,\Q_0,\V_0,\omega)$ 
would be close to $\widetilde \Q(t,\Q_0,\V_0,\omega)$. However, in 
Theorem~\ref{pathcomp}, we exclude a set of small probability that could depend 
on $t$. 

Hence one has to prove the phase space version of  Theorem~\ref{pathcomp}, which 
is easily done:  The vector fields to be treated in the phase space analog of 
Proposition~\ref{close} are now $(\widetilde\V,\F)$ and $(\widetilde 
\V,\widetilde \F)$.  But the position components are the same, and cancel 
identically, so the phase space analog of  Proposition~\ref{close} can be 
proved in exactly the same way, only with a more elaborate notation. 

Next, we may obtain an estimate for the Lyapunov exponent of the phase space 
flow from the one we obtained in Proposition~\ref{lyapunov}. To see this, note 
that since the phase space flow $\Xi_t$ is given by $$\Xi_t(\Q_0,\V_0) = 
\left(\Q_0 + \int_0^t \Psi_s(\V_0){\rm d} s, \Psi_t(\V_0)\right)\ ,$$ the 
Jacobian of $\Xi_t$ has the form $$D\Xi_t = \left[\begin{array}{cc}  I  & 
\int_0^t D\Psi_s{\rm d}s \\ 0 &D\Psi_t\end{array}\right]\ .$$ The crucial point 
here is that we have control on $\|D\Psi_s\|$ uniformly in $s$ for  all but a 
negligible set of paths. Thus we obtain control on  $\|D\Xi_t\|$ uniformly in 
$t$, except for a negligible set of paths. 

Form here the proof of Theorem~\ref{close} may be adapted to the phase space 
setting using these two estimates exactly as before, with only a more elaborate 
notation. 
\end{proof}

\begin{remark}  In Theorem~\ref{chaos}, we have proved that if the initial 
distribution is a {\em product} distribution, then at later times $t$ we have a 
chaotic distribution. This is somewhat less than propagation of chaos, though it 
is all that is needed to validate our Boltzmann equation. 

One can, however, adapt the proof to work with an chaotic sequence 
of initial distributions. We have used the product nature of the initial 
distribution to invoke the law of large numbers in two places:  Once in the  the 
proof of Proposition~\ref{close}, on the smallness of $\|\widetilde \F - \F\|$, 
and again in  the argument leading to (\ref{lln}).  However, the estimates we 
made only required {\em pairwise} independence of the velocities. Since  this 
is approximately true  for large $N$ for a chaotic sequence of distributions, 
by taking into account the small error term, one can carry out the same proof 
assuming only a chaotic initial distribution. 
\end{remark}

\begin{appendices}

\section{Properties of the Boltzmann equation.}\label{A.BE}

In this appendix, we prove existence and uniqueness for our Boltzmann equation. 
As in the the  proof of Theorem~\ref{chaos}, we do this first in the spatially 
homogeneous case, where the notation is less cumbersome, and then explain how 
to extend the argument to phase space.

As we have indicated, the key is the fact that the current satisfies an 
autonomous equation. Let $\tilde \j_0$ be computed from $f_0(\v)$, the initial 
data for (\ref{BE}). Then solve for $\tilde \j(t)$. Using this expression for 
$\tilde \j$ in (\ref{BE}) we obtain the 
equation
\begin{equation}\label{sh}
\frac{\partial }{\partial t}f(\v,t) + {\rm div}\left({\bf Y}(\v,t)f(\v,t)\right) = \mathcal{C}f(\v,t)\ ,
\end{equation}
where $\mathcal{C}$ denotes the collision operator on the right side of (\ref{BE}),
and where ${\bf Y}(\v,t)$ is the vector field
\begin{equation}\label{vecfie}
{\bf Y}(\v,t) =  \E-\frac{\E\cdot\tilde\j(t)}{\tilde u} \v\ ,
\end{equation}
which is linear, and hence globally Lipschitz.  Let $\Phi_t(\v)$ denote the flow 
corresponding to this vector field, so that $\v(t) := \Phi_t(\v)$ is the unique 
solution to $\dot \v(t) = {\bf Y}(\v(t),t)$ with $\v(0) = \v$.   Calling 
$g(\v,t)$ the push-forward of $f(\v,t)$ under the flow 
transformation $\Phi_t$, that is
$g(\v,t) = \det(D\Phi_t(\v))f(\Phi_t(\v),t)$, where $D\Phi_t$ is the Jacobian 
of $\Phi_t$, we deduce that $f(\v,t)$ satisfies (\ref{sh}) if and only if 
$g(\v,t)$ satisfies
\begin{equation}\label{sh2}
\frac{\partial }{\partial t}g(\v,t) = \widetilde {\mathcal{C}}_t g(\v,t)\ ,
\end{equation}
where
\[
\widetilde {\mathcal{C}}_tg(\v) = \det(D\Phi_{-t}) \left[ \mathcal{C} \left( 
\det(D\Phi_t) g\circ \Phi_t\right)\right]\circ \Phi_{-t}\ .
\]
Under our hypotheses, $\widetilde {\mathcal{C}}_t$ is bounded for each $t$, and 
depends continuously on $t$, and hence existence and uniqueness of solutions of 
(\ref{sh2}) is straightforward to prove. It is also straightforward to prove 
that the current generated by the solution $f(\v,t)$  satisfies the differential 
equation and initial conditions of Lemma \ref{curau} and hence it equals 
$\tilde \j$. 

To obtain the same result in the spatially inhomogeneous case, we replace the 
vector field ${\bf Y}$ in (\ref{vecfie}) by 
\begin{equation}
{\bf Y}(\q,\v,t) =  \left(\v\ ,\ \E-\frac{\E\cdot\tilde\j(t)}{\tilde u} \v\right)\ ,
\end{equation}
and let $\Phi_t$ be the phase space flow it generates. Again, the vector field 
is linear and hence globally Lipschitz which guarantees the existence, 
uniqueness and smoothness of flow. Now the same argument may be repeated.

\end{appendices}

\section*{Acknowledgments}
F. Bonetto, R Esposito and R. Marra gratefully acknowledge the hospitalty of Rutgers University, and E.A.  Carlen and J.L. Lebowitz
gratefully acknowledge the hospitality of  Universit\`a di Roma Tor Vergata.

\bibliographystyle{plain}
\bibliography{nonequi}

\begin{thebibliography}{10}

\bibitem{AM}
N.W. Ashcroft and N.D. Mermin.
\newblock Brooks Cole, 1983.

\bibitem{ba}
S.~Bastea, R.~Esposito, J.~L. Lebowitz, and R.~Marra.
\newblock Binary fluids with long range segregating interaction {I}: Derivation
  of kinetic and hydrodynamic equations.
\newblock {\em JSP}, 101:1087--1136, 2000.

\bibitem{BCKLs}
F.~Bonetto, N.~Chernov, A.~Korepanov, and J.~L. Lebowitz.
\newblock Speed distribution of {N} particles in the thermostated periodic
  {L}orentz gas with a field.
\newblock \url{http://arxiv.org/abs/1210.7720}, submitted to {\it {P}hys.
  {R}ev. {L}ett.}

\bibitem{BCKL}
F.~Bonetto, N.~Chernov, A.~Korepanov, and J.L. Lebowitz.
\newblock Spatial structure of stationary nonequilibrium states in the
  thermostated periodic {L}orentz gas.
\newblock {\em Journ, Stat. Phys.}, 146:1221--1243, 2012.

\bibitem{BCKL2}
F.~Bonetto, N.~Chernov, A.~Korepanov, and J.L. Lebowitz.
\newblock Nonequlibrium stationary state of a current-carrying thermostatted
  system.
\newblock {\em EPL}, 102:15001, 2013.

\bibitem{BDLR}
F.~Bonetto, D.~Daems, J.L. Lebowitz, and V.~Ricci.
\newblock Properties of stationary nonequilibrium states in the thermostatted
  periodic {L}orentz gas: The multiparticle system.
\newblock {\em Phys. Rev E}, 65:051204, 2002.

\bibitem{BrH}
W.~Braun and K.~Hepp.
\newblock The {V}lasov dynamics and its fluctuations in the 1/{N} limit of
  interacting classical particles.
\newblock {\em Comm. Math. Phys.}, 56:101--113, 1977.

\bibitem{CELS}
N.~Chernov, G.L. Eyink, J.L. Lebowitz, and Ya.G. Sinai.
\newblock Steady state electric conductivity in the periodic {L}orentz gas.
\newblock {\em Comm. Math. Phys.}, 154:569--601, 1993.

\bibitem{Doe}
J.~L. Doob.
\newblock Wiley, New York, 1953.

\bibitem{AG}
A.~Grunbaum.
\newblock Propagation of chaos for the {B}oltzmann equation.
\newblock {\em Arch. Rati. Mech. and Analysis}, 42:323--345, 1971.

\bibitem{Kac}
M.~Kac.
\newblock Foundations of kinetic theory.
\newblock In J.~Neyman, editor, {\em Proc. 3rd Berkeley Symp Math Stat. Prob},
  pages 171--197. Univ. of California, Vol 3, 1956.

\bibitem{La}
O.~Lanford.
\newblock Time evolution of large classical systems.
\newblock In J.~Moser, editor, {\em Dynamical Systems}, pages 1--97. 1975.

\bibitem{MH}
B.~Moran and W.~Hoover.
\newblock Diffusion in the periodic {L}orentz billiard.
\newblock {\em Journ, Stat. Phys.}, 48:709--726, 1987.

\bibitem{MM}
C.~Mouhot and S.~Mischler.
\newblock {K}ac’s program in kinetic theory.
\newblock {\em Inventiones mathematicae}, 2012.

\bibitem{Neu}
H.~Neunzert.
\newblock An introduction to the nonlinear {B}oltzmann-{V}lasov equation.
\newblock In C.~Cercignani, editor, {\em Lecture Notes in Mathematics Volume
  1048}, pages 80--100. 1984.

\bibitem{Sp}
H.~Spohn.
\newblock Kinetic equations from {H}amiltonian dynamics: {Ma}rkovian limits.
\newblock {\em Rev. Mod. Phys.}, 56:569, 1980.

\bibitem{Spb}
H.~Spohn.
\newblock {\em Large Scale Dynamics of Interacting Particles}.
\newblock Springer Verlag, 1991.

\bibitem{SA}
A.~Sznitman.
\newblock Equations de type {B}oltzmann, spatialment homogenes.
\newblock {\em Zeit. Wahr. und verwandte Geb.}, 66:559--592, 1984.

\bibitem{WW}
B.~Wennberg and Y.~Wondmagegne.
\newblock Stationary states for the {K}ac equation with a gaussian thermostat.
\newblock {\em Nonlinearity}, 17:633--648, 2004.

\end{thebibliography}

\end{document}